\documentclass[prl,twocolumn,showpacs,showkeys]{revtex4}

\usepackage{amssymb, amsmath, amsthm, amsfonts, latexsym, verbatim, mathrsfs}

\usepackage{graphicx}

\newtheorem{theorem}{Theorem}

\newcommand{\rA}{{\rm A}}
\newcommand{\rB}{{\rm B}}
\newcommand{\ra}{{\rm a}}
\newcommand{\rb}{{\rm b}}

\begin{document}

\title{Teleportation is necessary for faithful quantum state transfer \\ through noisy channels of maximal rank}
\thanks{The authors acknowledge support from the Emmy Noether Program of the DFG in Germany}

\author{Raffaele Romano}

\email{rromano@mpl.mpg.de}

\author{Peter van Loock}

\affiliation{Max-Planck Institute for the Science of Light, Institute of Theoretical Physics I
Universit\"{a}t Erlangen-N\"{u}rnberg, Staudstrasse 7/B2, 91058 Erlangen, Germany}

\begin{abstract}

\noindent Quantum teleportation enables deterministic and
faithful transmission of quantum states,
provided a maximally entangled state is pre-shared between
sender and receiver, and a one-way classical channel is available.
Here, we prove that these resources are not only sufficient,
but also necessary, for deterministically and faithfully sending quantum states through any fixed noisy
channel of maximal rank, when a single use of the cannel is admitted. In other words, for this family of channels,
there are no other protocols, based on different (and possibly cheaper) sets of resources, capable of replacing quantum
teleportation.

\end{abstract}



\maketitle


{\it Introduction.---}
It is a well established fact that information processing devices working at the quantum level
can outperform the corresponding classical devices~\cite{niel}. Prominent examples for this are
given for instance by quantum teleportation~\cite{benn2}, dense coding~\cite{benn4}, and quantum key distribution~\cite{benn5}.

A fundamental ingredient for the implementation of quantum technologies is the ability to faithfully
transmit quantum states of, say, an $N-$level system from a sender A to a receiver B. Typically, transmission
is accompanied by decoherence, and the system loses its special properties, because of its interaction with the external
environment. This irreversible process is described by a completely positive trace-preserving map
$\rho^{\rA} \rightarrow \rho^{\rB} = \varepsilon [\rho^{\rA}]$, where $\rho^{\rA}$ ($\rho^{\rB}$) is the $N \times N$
density matrix representing the state before (after) the transmission. In order to suppress the noise,
and achieve an improved transmission, some actions have to be performed by A and B, based on
the {\it resources} which are available to them. In the ideal case, a faithful and deterministic correction of the
channel is achieved, corresponding to turning $\varepsilon$ into the identity channel ${\mathcal I}$ with
unit probability.

The study of the aforementioned resources is fundamental both in the finite and asymptotic regimes. The latter
case, where an arbitrarily large number of uses of $\varepsilon$ is admitted, is highly relevant from an information
theoretical perspective (e.g., for the derivation of coding theorems)~\cite{deve}. However, practical implementations
always rely on the finite regime; therefore, we will limit our attention to this case.
In this context, several schemes have been proposed in the past two decades, basically falling into two
large families: Quantum Error Correction Codes (QECCs)
and protocols relying on Quantum Teleportation (QT).

The essence of QECCs is that the state to be transmitted is encoded in a higher-dimensional system,
which is then sent through several independent channels. The state is finally recovered by means of
suitable measurements and operations at the receiving station. QECCs exist that provide {\it universal}
protection of the transmitted state, that is, protection against an arbitrary channel $\varepsilon$~\cite{shor,knil,benn}.
In QT, transmission is based on a maximally entangled pair shared between sender and receiver, and on classical
communication~\cite{benn2}. This protocol will be shortly reviewed. Since the physical channel described by $\varepsilon$
is not used at all (except ``off-line", for entanglement distribution and distillation~\cite{benn3}), the protocol
automatically provides universal protection of the transmitted state.

Protecting the transmission of quantum states through noisy channels, by using either QECCs or QT, is
a difficult task, due to, in one case, the implementation of operations on multiple systems, and in the other case,
the generation (and preservation) of bipartite maximal entanglement. However, when the details of the physical channel are
known, one has to deal with a {\it specific} noisy action, and universal protocols are not needed. In this
case, a reduction of the resources needed for faithful and deterministic transmission is expected. For example,
the dimension of the enlarged spaces of QECCs is lowered with respect to the general case when dealing with
specific noisy channels.

Given a general set of local and non-local resources, and a fixed channel $\varepsilon$,
it is thus of interest to answer the following question: is it possible to construct protocols
which perform as well as QT, but nonetheless are different from QT (in particular, including the
use of the physical channel)? More generally, what are the {\it necessary} resources for
deterministically and faithfully protecting quantum states sent through $\varepsilon$?

As a first step in this analysis, in this work we prove that QT is the only protocol satisfying our requirements,
assuming that: (i) $\varepsilon$ has maximal rank $N^2$~\footnote{By definition, the rank of $\varepsilon$ is
minimal number of Kraus operator needed to represent the channel in the Operator Sum Representation},
and (ii) $\varepsilon$ can be used only once. Therefore, the resources of QT are {\it necessary and
sufficient} for perfect transmission when a single use of a maximally-ranked channel is admitted, and a
cheaper set of resources does not exist.

There are situations in which this result is not surprising. Assume for instance that
$\varepsilon$ is the $p-$depolarizing channel for a single qubit,
\begin{equation}\label{comdepcha}
    \varepsilon [\rho^{\rA}] = p \frac{I}{2} + (1 - p) \rho^{\rA},
\end{equation}
with $\frac{2}{3} \leqslant p \leqslant 1$, where $I$ is the $2 \times 2$ identity matrix.
It is known that such a noisy channel can be simulated by local actions and classical communication,
and then it doesn't help to send the state through it~\footnote{For $p = \frac{2}{3}$, $\varepsilon$ is the classical
limit of teleportation}. Since for disembodied faithful and deterministic transmission of a quantum state, maximal entanglement shared between
A and B is necessary~\cite{benn2}, we get the expected result. Nevertheless, this argument no longer holds
for other channels, for example when $p < \frac{2}{3}$ in (\ref{comdepcha}), where
the use of $\varepsilon$ might, in principle, reduce the required amount of shared entanglement.
Moreover, it does not hold whenever classical communication is not a priori considered a free resource.
Therefore, it is justified to analyze what are the minimal resources needed for the perfect correction of channels
of maximal rank.


{\it Quantum teleportation.---} Let us briefly outline the QT scheme in its standard
form \cite{benn2}. Upper A (B) denotes quantities referring to the input
(output) of the channel, as well as the sender and the receiver; the corresponding small
letters label the respective auxiliary systems. The two parties
share the maximally entangled state $\Psi_0^{\ra \rb} = \vert \psi_0 \rangle^{\ra \rb} \langle \psi_0 \vert$, where
\begin{equation}\label{bell}
    \vert \psi_0 \rangle^{\ra \rb} = \frac{1}{\sqrt{N}} \sum_{i = 0}^{N - 1} \vert i \rangle^{\ra} \otimes \vert i \rangle^{\rb}.
\end{equation}
The sender performs a Bell measurement on the space labeled by A and a,
defined by the set of projectors $\Psi_{\eta}^{\rA \ra} = \vert \psi_{\eta} \rangle^{\rA \ra} \langle \psi_{\eta} \vert$, where
\begin{equation}\label{bell}
    \vert \psi_{\eta} \rangle^{\rA \ra} = \frac{1}{\sqrt{N}} \sum_{k = 0}^{N - 1} e^{2 \pi i k \frac{n}{N}}
    \vert k \rangle^{\rA} \otimes \vert (k + m)_N \rangle^{\ra},
\end{equation}
$\eta = 0, \ldots, N^2 - 1$, $n = \eta \,{\rm div}\, N$,
$m = \eta \,{\rm mod}\, N$ (that is, $\eta = n N + m$), and we used the notation $(i)_N = i \,{\rm mod}\, N$.
Then, A sends the measurement outcome $\eta$ to B through a noiseless classical channel.
According to this result, B applies the unitary operator given by
\begin{equation}\label{unitqt}
    U_{\eta}^{\rB \rb} = U^{\rB \rb}_{swap} \sum_{k = 0}^{N - 1} e^{2 \pi i k \frac{n}{N}} I^{\rB} \otimes \vert k \rangle^{\rb}
    \langle (k + m)_N \vert,
\end{equation}
reproducing the unknown initial state on his side.
The swap operator $U_{swap}^{\rB \rb}$ and the identity $I^{\rB}$ have been introduced such that
the output appears in the output of the channel rather than in the corresponding
auxiliary system; we prefer to discard the ancillae at the end of
the procedure. The complete protocol is thus described by
\begin{equation}\label{telepst}
    \rho^{\rB} = \sum_{\eta = 0}^{N^2 - 1} {\rm Tr}_{\ra \rb} \Bigl( U^{\rB \rb}_{\eta} \varepsilon [ \Psi^{\rA \ra}_{\eta} \rho^{\rA}
    \otimes \Psi_0^{\ra \rb} \Psi^{\rA \ra}_{\eta}] U_{\eta}^{\rB \rb \dagger} \Bigr),
\end{equation}
and it turns out that $\rho^{\rB} = \rho^{\rA}$ for all states $\rho^{\rA}$, and all noisy channels $\varepsilon$.

The first experimental realizations of the QT protocol were reported in~\cite{bosc,zeil}.
In general, the implementation of this protocol is difficult, as it requires non-local resources.
In any real setting, to share a maximally entangled state is demanding. If the entanglement is not maximal,
the standard QT protocol provides unfaithful transmission~\cite{benn2}; an arbitrary mixed state used as entangled
resource for the teleportation results in state depolarization~\cite{bowe}. Several schemes for approaching faithful transmission have
been proposed, prominent examples being entanglement distillation~\cite{benn3}, conclusive QT~\cite{mor,son},
and multiple QT~\cite{modl}. However, none of these protocols is deterministic.

We now consider the most general protocol including transmission through the physical channel $\varepsilon$, to
compare the resources it requires for deterministic and faithful transmission with those corresponding to QT.


{\it Impact of the available resources on noisy channels.---}
We assume that A and B can perform the following operations (which we also call {\it resources}): a) attach auxiliary
systems (ancillae), compactly represented by $\rho^{\ra \rb}$, to the main system, or discard them; b) perform local
operations (unitary transformations and projective measurements) on their composite systems; c) exchange classical
communication via a noiseless classical channel. This set of operations  is generically denoted by LOCC when $\rho^{\ra \rb}$
is a separable state. In particular, arbitrary completely positive local operations or generalized measurements (POVMs)
can be performed by A and B in this way; classical communication can be used to correlate the operations performed by A
and B. Pre-shared entanglement is accounted for through an entangled global state $\rho^{\ra \rb}$. Entanglement provides
a non-classical correlation between the operations performed by A and B.

Usually, LOCC operations are considered easy to perform, and it is assumed that there is no limitation in their
implementation. In what follows, we adopt a more general point of view, and assume that only uncorrelated local
operations are easily implementable~\footnote{Usually, classical communication is easily implementable; however,
it could imply additional, expensive resources. For example, in long-distance quantum communication via quantum repeaters,
more classical communication requires longer waiting times, and thus for better memories~\cite{hart}}. 
In this way, we are able to find constraints referring to both shared entanglement
and classical communication.

We argue that we can limit our attention to one-way classical communication without loss of generality. In fact,
any measurement performed on one side can be simulated, at least in principle, by a measurement performed on the
other side, followed by local operations~\cite{lo}. Therefore, only the measurement outcome obtained by A before
the state transmission is essential. The composite action of resources a) - c) leads to the state evolution
\begin{equation}\label{telepstgen}
    \rho^B = \tilde{\varepsilon} [\rho^{\rA}] = \sum_{\eta = 0}^{M - 1} {\rm Tr}_{\ra \rb} L_{\eta \eta^{\prime}}^{\rB \rb} \varepsilon
    [L_{\eta}^{\rA \ra} \rho^{\rA} \otimes \rho^{\ra \rb} L_{\eta}^{\rA \ra \dagger}] L_{\eta \eta^{\prime}}^{\rB \rb \dagger},
\end{equation}
where $L^{\rA \ra}_{\eta}$ and $L^{\rB \rb}_{\eta \eta^{\prime}}$ are the local operations performed on the A and B sides,
including ancillae, and $\eta$, $\eta^{\prime}$ label different measurement outcomes in A and B respectively. In full generality,
$L^{\rA \ra}_{\eta} = \Pi^{\rA \ra}_{\eta} U^{\rA \ra}_{\eta}$, where $U^{\rA \ra}_{\eta}$ are unitary operators, $\Pi^{\rA \ra}_{\eta}$
projections, and similar relations holds for $L^{\rB \rb}_{\eta \eta^{\prime}}$. The state $\rho^{\ra \rb}$ is
assumed to be pure~\footnote{Mixed states can be included by linearity in the general treatement}, $\rho^{\ra \rb} = \vert \psi \rangle^{\ra \rb} \langle \psi \vert$, with Schmidt decomposition
\begin{equation}\label{schmidt}
    \vert \psi \rangle^{\ra \rb} = \sum_{i = 0}^{P - 1} \mu_i \vert i \rangle^{\ra} \otimes \vert i \rangle^{\rb}.
\end{equation}
The Schmidt coefficients $\mu_i \geqslant 0$, with $\sum_i \mu_i^2 = 1$, account for pre-existing
entanglement shared between A and B, prepared off-line before transmitting the signal state; $P$ is the
dimension of the individual systems owned by sender and receiver. Finally,
the index $\eta$ labels $M$ different operations, correlated using classical communication.
Note that QT, as in (\ref{telepst}), is a special case of (\ref{telepstgen}).

The impact of all these resources can be interpreted as a form of control over the noisy channel $\varepsilon$,
with the ultimate goal to obtain the identity channel for a suitable choice of the operations a) - c). Equation (\ref{telepstgen})
defines the map $\varepsilon \rightarrow \tilde{\varepsilon} = \lambda [\varepsilon]$, with the target
$\tilde{\varepsilon} = {\mathcal I}$. This map is the mathematical representation of the protocol
realized through the specific resources that have been considered. To characterize this map, it is convenient to use the Choi-Jamiolkowski
isomorphism between completely positive trace-preserving maps and positive, unit-trace operators $\varepsilon \simeq R^{\rB \rA}$ acting
on the joint space of B and A \cite{jami}, defined by the action of the channel on half of the maximally entangled state,
\begin{equation}\label{choi}
    R^{\rB \rA} = \varepsilon \otimes I^{\rA} [\Psi_0^{\rA \rA}].
\end{equation}
Using this isomorphism, the control action is expressed by the linear, completely positive map
\begin{equation}\label{control}
    \tilde{R}^{\rB \rA} = \lambda [R^{\rB \rA}] = \sum_{\eta = 0}^{M - 1} \sum_{k, l = 0}^{N^2 - 1}
    \Lambda^{\eta}_{k, l} R^{\rB \rA} \Lambda^{\eta \dagger}_{k, l},
\end{equation}
with operators $\Lambda^{\eta}_{k,l}$ that depend on the aforementioned resources as~\footnote{For simplicity, we denote by
$\lambda$ both the map acting on $\varepsilon$ and the map acting on the corresponding $R^{\rB \rA}$, although they are different},
\begin{equation}\label{kraus}
    \Lambda^{\eta}_{k, l} = \sum_{i = 0}^{N^2 - 1} \mu_i B_{k,i}^{\eta} \otimes A^{\eta T}_{l,i}.
\end{equation}
In particular,
\begin{equation}\label{local}
    A_{i,j}^{\eta} = \langle i \vert^{\ra} L_{\eta}^{\rA \ra} \vert j \rangle^{\ra}, \quad B_{i,j}^{\eta} = \langle i
    \vert^{\rb} L_{\eta \eta^{\prime}}^{\rB \rb} \vert j \rangle^{\rb},
\end{equation}
and the bases $\{ \vert i \rangle^{\ra}, i \}$ and $\{ \vert i \rangle^{\rb}, i \}$ are defined by the Schmidt decomposition (\ref{schmidt}).
In general, the map $\lambda$ is non-trace preserving (more details can be found in~\cite{roma}).

Using this representation, the desired channel manipulation $\varepsilon \rightarrow \tilde{\varepsilon} = {\mathcal I}$ reads
\begin{equation}\label{theooothe}
    \sum_{\eta} \sum_{k, l} \Lambda^{\eta}_{k, l} R^{\rB \rA} \Lambda^{\eta \dagger}_{k, l} = \Psi_0^{\rB \rA}.
\end{equation}
By this approach, we mapped the problem into a geometric control problem of steering an initial
state to a desired target state through a protocol $\lambda$.


{\it Faithful and deterministic transmission through $\varepsilon$.---}
With the formalism introduced so far, we can derive the main result of this work. We stress that
it is assumed that operations of type b) can be performed without limitations; we shall analyze the
constraints over classical communication and shared entanglement.

\begin{theorem}\label{teorema}
The resources of QT, i.e., classical communication and maximal entanglement,
are necessary and sufficient to deterministically map a trace-preserving noisy channel $\varepsilon$
with maximal rank $N^2$ into $\mathcal I$. In other word, they are necessary and sufficient to provide
deterministic and faithful transmission through $\varepsilon$.
\end{theorem}

\begin{proof}
Sufficiency follows from the universality of the QT protocol. Therefore, we only need to prove that
classical communication and maximal entanglement are necessary for the desired task. Notice that
a deterministic protocol necessarily requires $\sum_{\eta} L_{\eta}^{\rA \ra \dagger} L_{\eta}^{\rA \ra} =
\sum_{\eta} L_{\eta}^{\rA \ra} L_{\eta}^{\rA \ra \dagger} = I^{\rA \ra}$, and $L_{\eta \eta^{\prime}}^{\rB \rb}$
independent of $\eta^{\prime}$ and unitary for all $\eta$. Therefore, the operators in (\ref{local}) must
satisfy
\begin{eqnarray}\label{rela}
  \sum_{\eta} \sum_k A^{\eta}_{i,k} A_{j,k}^{\eta \dagger} &=& \sum_{\eta} \sum_k A_{k,i}^{\eta \dagger} A^{\eta}_{k,j} = \delta_{ij} I^{\rA},
 \nonumber \\
  \sum_k B^{\eta}_{i,k} B_{j,k}^{\eta \dagger} &=& \sum_k B_{k,i}^{\eta \dagger} B^{\eta}_{k,j} = \delta_{ij} I^{\rB}.
\end{eqnarray}
To prove that classical
communication is needed, we first assume this is not the case by imposing $M = 1$ and dropping the label $\eta$,
looking for a contradiction. Consider $R \simeq \varepsilon$, and write its spectral decomposition as
$R = \sum_i r_i \vert r_i \rangle^{\rB \rA} \langle r_i \vert$. Notice that $r_i > 0$ for all $i = 0, \ldots, N^2 - 1$,
since $\varepsilon$ has maximal rank. The protocol $\lambda $ satisfies (\ref{theooothe}) if and only if it
deterministically maps $\vert r_i \rangle^{\rB \rA} \langle r_i \vert$ into $\Psi_0^{\rB \rA}$ for all $i = 0, \ldots, N^2 - 1$.
Since $\{\vert r_i \rangle^{\rB \rA}, i \}$ is a complete set, it can be replaced by an arbitrary orthonormal set $\{
\vert v_i \rangle^{\rB \rA}, i\}$, and it is always possible to write
\begin{equation}\label{ort}
    \Lambda_j \vert v_i \rangle^{\rB \rA} = \delta_{ij} \vert \psi_0 \rangle^{\rB \rA},
\end{equation}
by suitably redefining the operators $\Lambda_j$. The index $j$ embodies all the needed indices appearing in
(\ref{kraus}). If we choose a set of factorized vectors, $\vert v_i \rangle^{\rB \rA} = \vert m \rangle^B \otimes
\vert n \rangle^A$, from (\ref{ort}) we obtain
\begin{equation}\label{conimp}
    \sum_i \mu_i \langle n \vert A_{l,i} B_{k,i} \vert m \rangle = \sqrt{N} \, \delta_{km} \delta_{ln},
\end{equation}
where unnecessary upper indices have been omitted. We multiply (\ref{conimp}) by its complex conjugate,
sum over the indices $n, k, l$, and use (\ref{rela}), and finally find the contradiction $\sum_i \mu_i^2 = N P$.
We conclude that classical communication is a needed resource, so $M > 1$.

Restoring the label $\eta$, and proceeding as before, we find that
\begin{equation}\label{conimp2}
    \sum_i \mu_i \langle n \vert A^{\eta}_{l,i} B^{\eta}_{k,i} \vert m \rangle = \sqrt{N} \,
    \beta_{\eta} \, \delta_{km} \delta_{ln}.
\end{equation}
The l.h.s. of (\ref{conimp2}) can be seen as an inner product; by using the Cauchy-Schwarz
inequality, we can write
\begin{equation}\label{schwar}
    \sum_{ijpq} \mu_i \mu_p \vert \langle n \vert A_{li}^{\eta} \vert j \rangle \vert^2
    \vert \langle m \vert B_{kp}^{\eta \dagger} \vert q \rangle \vert^2 \leqslant
    N \, \vert \beta_{\eta} \vert^2 \delta_{km} \delta_{ln}.
\end{equation}
Now we sum over the indices $k, l, n$ and $\eta$, and use the properties (\ref{rela}). Finally, we obtain
$\sum_i \mu_i \geqslant \sqrt{N}$. This relation can be satisfied only if $P \geqslant N$.
For $P = N$, it implies $\mu_i = \frac{1}{\sqrt{N}}$ for all $i$, that is, maximal entanglement.
For $P > N$, since CC has been proven to be necessary, we can use the theorem by Nielsen~\cite{niel2}
characterizing the transitions among pure states under LOCC for bipartite systems. It follows
that it is always possible to manipulate $\rho^{\ra \rb}$ via LOCC, such that $\mu_i = \frac{1}{\sqrt{N}}$ for
$i = 0, \ldots, N - 1$, and the remaining $P - N$ coefficients vanish. Therefore any protocol
needs the resources of QT.
\end{proof}


{\it Conclusions.---} In this work, we have proven that QT is the only protocol that
can faithfully and deterministically correct a noisy channel $\varepsilon$ of maximal rank,
affecting a finite-dimensional system.
To prove this result, we have used a formalism, based on the Choi-Jamiolkowski
isomorphism, providing a connection between the study of noisy channels and
geometric control theory. From this point of view, the desired task corresponds to
mapping an arbitrary Choi-Jamiolkowski state to the Bell state representing the
perfect channel ${\mathcal I}$.

Our results imply that, as long as a single use of the channel is admitted, maximal
entanglement and classical communication between sender and receiver are minimal
resources. This is true even if a physical (full-rank) quantum channel is available, that affects the 
transmission arbitrarily little, while only strictly faithful quantum state transfer is allowed.
A similar situation in encountered in QECCs, where the needed resources for encoding depend on the structure
of the errors, and not on their magnitude. It is an open question whether, as for QECCs,
lowering the rank of $\varepsilon$  might result in a reduction of
the needed resources for the implementation of faithful and deterministic protocols.

\vfill


\end{document}